\documentclass{article}
\usepackage{graphicx,amssymb,amsmath}
\usepackage{url,float}
\usepackage{amsfonts}
\usepackage{latexsym}
\usepackage[boxruled]{algorithm2e}
\usepackage{tcolorbox} 

\usepackage{color}

\newcommand{\remove}[1]{{}}

\usepackage[normalem]{ulem}

\newcommand{\ABox}{
\raisebox{3pt}{\framebox[6pt]{\rule{6pt}{0pt}}}
}
\newenvironment{proof}{{\bf Proof:}}{\hfill\ABox}

\newtheorem{theorem}{{\bf Theorem}}

\newtheorem{lemma}{Lemma}

\newcommand{\lemlab}[1]{\label{lemma:#1}}
\newcommand{\thmlab}[1]{\label{thm:#1}}

\newcommand{\figlab}[1]{\label{fig:#1}}
\newcommand{\seclab}[1]{\label{sec:#1}}

\newcommand{\lemref}[1]{\ref{lemma:#1}}
\newcommand{\thmref}[1]{\ref{thm:#1}}

\newcommand{\secref}[1]{\ref{sec:#1}}
\newcommand{\figref}[1]{\ref{fig:#1}}

\def\P{{\mathcal P}}
\def\C{{\mathcal C}}
\def\bcC{{\partial \mathcal C}}

\def\g{{\gamma}}
\def\f{{\phi}}
\def\F{{\Phi}}
\def\l{{\lambda}}
\def\r{{\rho}}
\def\o{{\omega}}
\def\d{{\delta}}
\def\e{{\varepsilon}}
\def\R{{\mathbb{R}}}

\newcommand{\squeezelist}{\setlength{\itemsep}{0pt}}

\title{Un-unzippable Convex Caps%
\thanks{The title mimics
that of the paper~\protect\cite{bdekms-upcf-03}:
``Ununfoldable Polyhedra with Convex Faces''
}}

\author{
Joseph O'Rourke%
    \thanks{Department of Computer Science, Smith College, Northampton, MA, USA.
      \protect\url{jorourke@smith.edu}.}
}

\index{O'Rourke, Joseph}

\begin{document}
\maketitle

\begin{abstract}
An unzipping of a polyhedron $\P$ is a cut-path through its vertices
that unfolds $\P$ to a non-overlapping shape in the plane.
It is an open problem to decide if every convex $\P$ has an unzipping.
Here we show that there are nearly flat convex caps that
have no unzipping.
A convex cap is a ``top'' portion of a convex polyhedron;
it has a boundary, i.e., it is not closed by a base.
\end{abstract}

\section{Introduction}
\seclab{Introduction}
We define an \emph{unzipping} of a polyhedron $\P$ in $\R^3$ to be a 
non-overlapping, single-piece unfolding
of the surface to the plane that results from cutting a continuous path 
$\g$ through all the vertices of $\P$. 
The cut-path $\g$ need not follow the edges of $\P$, nor even
be polygonal, but it must include
every vertex of $\P$, passing through $n-2$ vertices and beginning and
ending at the other two vertices, where $n$ is the total number of vertices.
If $\g$ does follow edges of $\P$, we call it an \emph{edge-unzipping}.

Edge-unzippings are special cases of edge-unfoldings, where the cuts follow
a tree of edges that span the $n$ vertices.
The interest in edge-unfoldings stems largely from what has
become known as D\"urer's problem~\cite{do-gfalop-07}~\cite{o-dp-13}:
Does every convex polyhedron have an edge-unfolding?
The emphasis here is on a non-overlapping result, what is often
called a \emph{net} for the polyhedron.
This question was first formally raised by Shephard in~\cite{s-cpcn-75}.
In that paper, he already investigated the special case where the
cut edges form a Hamiltonian path of the $1$-skeleton of $\P$:
Hamiltonian unfoldings. These are exactly what I'm calling edge-unzippings.
Shephard noted that the rhombic dodecahedron does
not have an edge-unzipping because its $1$-skeleton has no
Hamiltonian path.

The attractive ``zipping'' terminology stems from the paper~\cite{lddss-zupc-10},
which defined \emph{zipper unfoldings} to be what I'm shortening to unzippings.
They showed that all the Platonic and the Archimedean solids have
edge-unzippings.
And they posed a fascinating question: 

\begin{center}
\fbox{\textbf{Open Problem}: Does every convex polyhedron have an unzipping?}
\end{center}

\subsection{Nonconvex Polyhedra}
\seclab{NonconvexPolyhedra}
First we note that not every nonconvex polyhedron has an unzipping.
This has been a ``folk theorem'' for years, but has apparently not been
explicitly stated in the literature.\footnote{
The closest is~\cite{bdekms-upcf-03}, which notes that
``the neighborhood of a negative-curvature vertex ... requires two or more
cuts to avoid self-overlap.''}
In any case, it is not difficult to see.

Consider the polyhedron illustrated in Fig.~\figref{heartsfan}.
The central vertex $v$ has more than $4 \pi$ incident surface angle.
In fact, it has well more than $8 \pi$ incident angle, but we only need $> 4 \pi$.
An unzipping cut-path $\g$ cannot terminate at $v$, because the neighbhood
of $v$ in the unfolding has more than $2\pi$ incident angle, and so would
overlap in the planar development.
Nor can $\g$ pass through $v$, because partitioning the $> 4 \pi$ angle would
leave more than $2\pi$ to one side or the other, again forcing overlap
in the neighborhood of at least one of the two planar images of $v$.
Therefore, no polyhedron with a vertex with more than $> 4 \pi$ incident angle
has an unzipping.
Indeed, as Stefan Langerman observed,\footnote{
Personal communication, Aug. 2017} 
similar reasoning shows that for
any degree $\d$ there is a polyhedron that cannot be unfolded without overlap
by a cut tree of maximum degree $\d$.
The polyhedron in Fig.~\figref{heartsfan} requires degree $> 4$ at $v$
to partition the more than $8 \pi$ angle into $< 2\pi$ pieces.
\begin{figure}[htbp]
\centering
\includegraphics[width=0.4\linewidth]{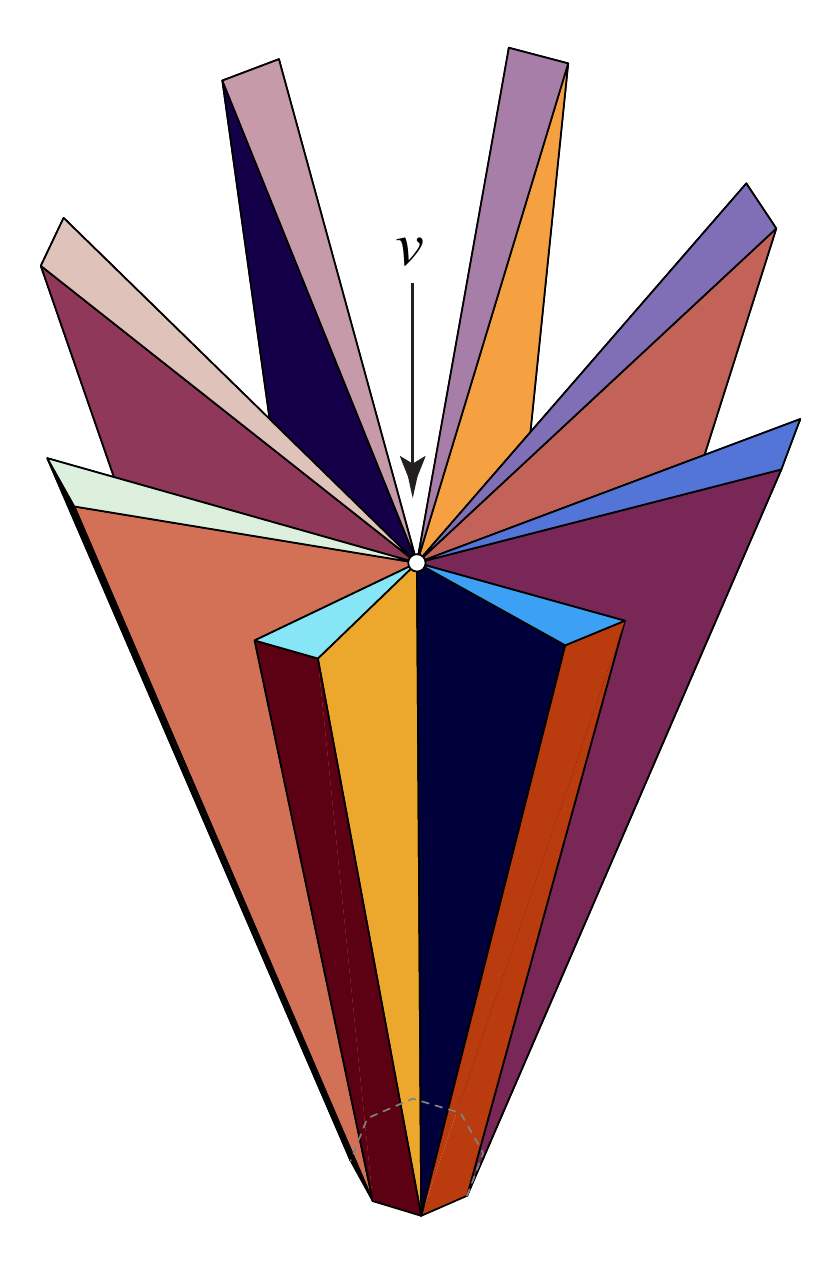}
\caption{A polyhedron that cannot be unzipped.
Based on Fig.~24.14, p.370 in~\protect\cite{do-gfalop-07}.}
\figlab{heartsfan}
\end{figure}

\subsection{Open Problem: Conjecture}
This negative result for nonconvex polyhedra increases the interest in the open problem for convex polyhedra.
In~\cite{o2015spiral} I conjectured the answer is {\sc no}, but it seems
far from clear how to settle the problem.
For that reason, here we turn to a very special case.

\subsection{Convex Caps}
\seclab{ConvexCaps}
The special case is unzipping ``convex caps.''
I quote the definition from~\cite{o-eunfcc-17}:
\begin{quotation}
\noindent
``Let $\P$ be a convex polyhedron, and let $\phi(f)$  be
the angle the normal to face $f$ makes with the $z$-axis.
Let $H$ be a halfspace whose bounding plane is orthogonal to the $z$-axis, and includes points
vertically above that plane.
Define a \emph{convex cap} $\C$ of angle $\F$ to be $C=\P \cap H$
for some $\P$ and $H$, such that $\f(f) \le \F$ for all $f$ in $\C$. [...]
Note that $\C$ is not a closed polyhedron; it has no ``bottom,''
but rather a boundary $\bcC$.''
\end{quotation}

\noindent
The result of this note is:
\begin{theorem}
For any $\F > 0$, there is a convex cap $\C$ that has no unzipping.
\thmlab{ZipCex}
\end{theorem}
\noindent
Because this holds for any $\F > 0$, 
there are arbitrarily flat convex caps that cannot be unzipped.
($\F$ will not otherwise play a role in the proof.)

\section{Proof of Theorem~1}
\seclab{proof}
The convex caps used to prove the theorem are all variations on the cap
shown in Fig.~\figref{ZipCex3D}.
The base $\bcC = (b_1,b_2,b_3)$ forms a unit side-length equilateral triangle in the $xy$-plane.
The three vertices $a_1,a_2,a_3$ are also the corners of an equilateral triangle,
lifted a small amount $z_a$ above the base.
In projection to the the $xy$-plane, the ``apron'' of quadrilaterals between
$\triangle b_1 b_2 b_3$ and $\triangle a_1 a_2 a_3$ has width $\e > 0$.
The vertex $c$, at height $z_c > z_a$, sits over the centroids of the equilateral triangles.
The shape of the cap is controlled by three parameters: $\e, z_a, z_b$.
Keeping $\e$ fixed and varying $z_a$ and $z_c$ permits controlling the
curvatures $\o_a$ at $a_i$ and $\o_c$ at $c$.
In Fig.~\figref{ZipCex3D}, $\e=0.1$
and $z_a,z_c = 0.02, 0.1$
leads to $\o_a=1.9^\circ$ and $\o_c=5.6^\circ$.
\begin{figure}[htbp]
\centering
\includegraphics[width=0.75\linewidth]{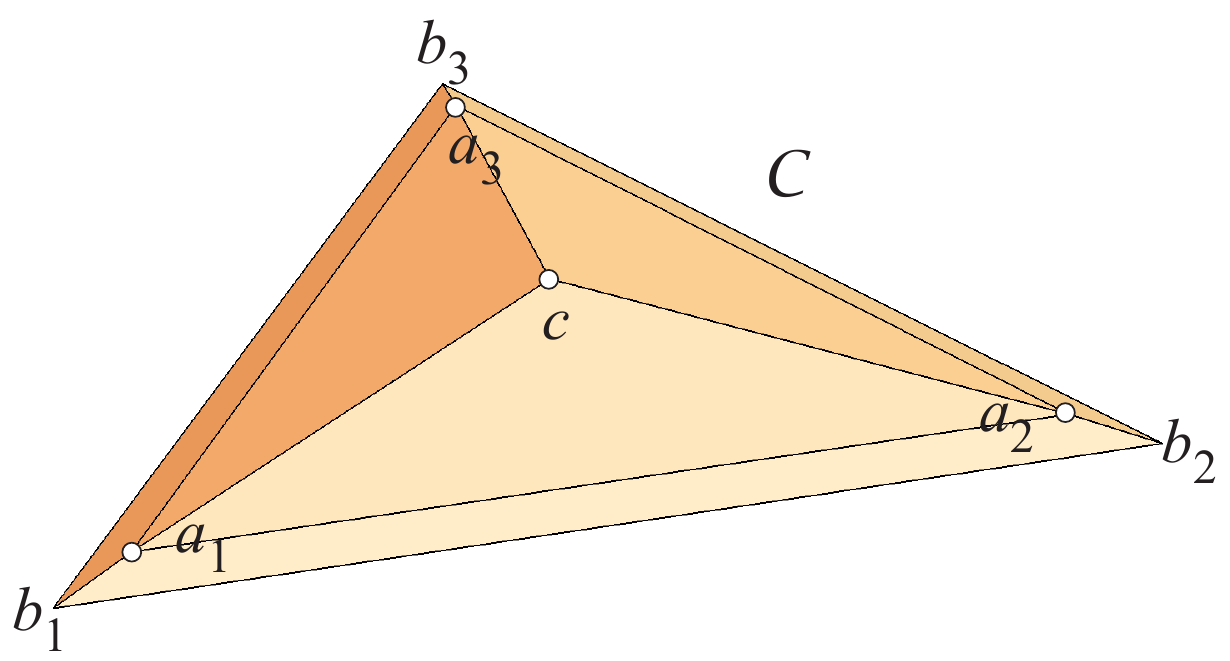}
\caption{A convex cap $\C$ that has no unzipping.}
\figlab{ZipCex3D}
\end{figure}

A typical attempt at an unzipping (of a variant of
Fig.~\figref{ZipCex3D}) is shown in Fig.~\figref{Layout_caaa}.
In general we will only display what are labeled $L$ and $R$ in this figure, rather
than the full unfolding.
\begin{figure}[htbp]
\centering
\includegraphics[width=0.6\linewidth]{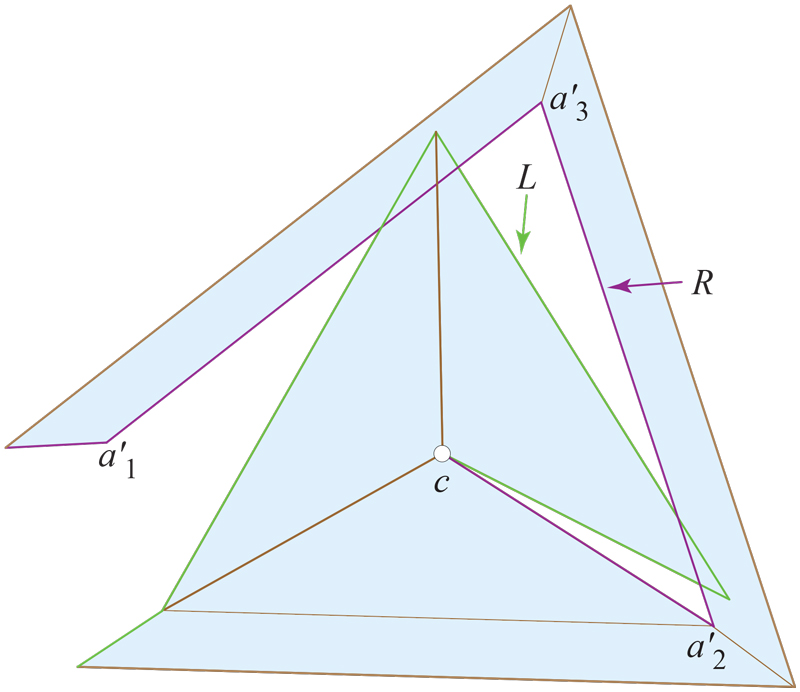}
\caption{An overlapping unfolding of a convex cap 
(a variant of Fig.~\protect\figref{ZipCex3D})
from cut-path
$\g =(c, a_2, a_3, a_1, b_1)$.
Compare Fig.~\protect\figref{Devel_caaa_linear_5_10} ahead.}
\figlab{Layout_caaa}
\end{figure}
From now on we will illustrate cut-paths and unzippings in the plane, starting
from Fig.~\figref{EqTris} (and not always repeating all the labels).
\begin{figure}[htbp]
\centering
\includegraphics[width=0.5\linewidth]{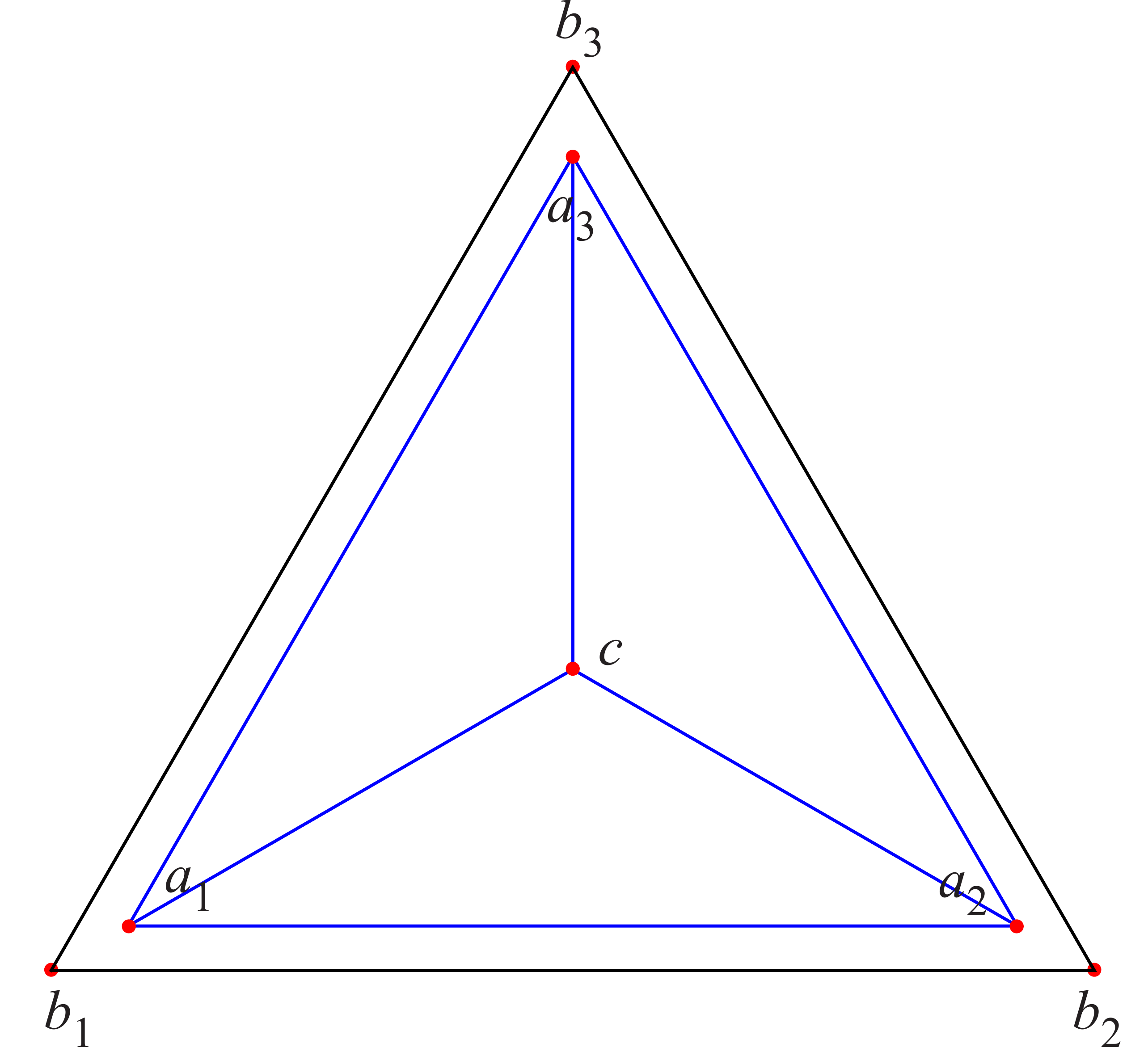}
\caption{Projection of Fig.~\protect\figref{ZipCex3D} to $xy$-plane.}
\figlab{EqTris}
\end{figure}

\subsection{Constraints on the cut-path}
\seclab{Constraints}
Any point $p$ in the relative interior of $\g$ (i.e., not an endpoint)
develops in the plane to two points $p'$ and $p''$,
with right and left incident surface angles $\r=\r(p)$ and $\l=\l(p)$.
If $p$ is not at a vertex of $\C$, then $\l+\r=2\pi$.
If $p$ is at a vertex of curvature $\o=\o(p)$, then $\l+\r+\o=2\pi$.
We will show the development $\C$ as cut by $\g$ by drawing two directed paths
$R$ and $L$, each determined by the $\r$ and $\l$ angles,
which deviate by $\o(v)$ at each vertex $v \in \g$.
The surface of $\C$ is right of $R$ and left of $L$
(see Fig.~\figref{Layout_caaa}), but not explicitly
depicted in subsequent figures.

The constraints on $\g$ to be an unzipping are:
\begin{enumerate}
\squeezelist
\item $\g$ must be a path, by definition of unzipping.
\item $\g$ must start at one of the vertices $\{c,a_1,a_2,a_3\}$
and terminate on $\bcC$.
\item $\g$ does not have to include any of the vertices $\{b_1,b_2,b_3\}$,
it just needs to exit $\C$ at some point of $\bcC$.
\item $\g$ can only touch $\bcC$ at one point, for if it touches at two
or more points, the unfolding would be disconnected into more than one piece.
\item Between vertices, $\g$ can follow any path on $\C$, as long as $\g$ 
does not self-cross, which would again result in more than one piece.
\item And of course, the developments of $R$ and $L$ must not cross in the plane,
for $R / L$ crossings imply overlap.\footnote{
The reverse is not always true: It could be that $R$ and $L$ do not cross, but
other portions of the surface away from the cut $\g$ are forced to overlap
by, for example, large curvature openings.}
\end{enumerate}

We think of $\g$ as directed from its root start vertex
to $\bcC$; the path opens from the root to its boundary exit.
The main constraint we exploit is item~4: $\g$ can only touch $\bcC$ at one point.
We will see that only by leaving $\C$ and returning could the unzipping avoid overlap.

Due to the symmetry of $\C$---in particular, the equivalence of $\{a_1,a_2,a_3\}$---there are only four combinatorially distinct possible cut-paths $\g$,
where we use $b$ to represent any point on $\bcC$:
\begin{enumerate}
\squeezelist
\item $\g = (c, a_1, a_2, a_3, b) = caaab$.
\item $\g = (a_1, c, a_2, a_3, b) = acaab$.
\item $\g = (a_1, a_2, c, a_3, b) = aacab$.
\item $\g = (a_1, a_2, a_3, c, b) = aaacb$.
\end{enumerate}
We abbreviate the path structure with strings $acaab$ and so on, with the obvious meaning.
It turns out that the location of $b$, the point at which $\g$ exits $\C$, plays
little role in the proof.

We will display the structure of $\g$ and the developments of $R$ and $L$
as in
Fig.~\figref{Path_acaa_linear}.
\begin{figure}[htbp]
\centering
\includegraphics[width=0.75\linewidth]{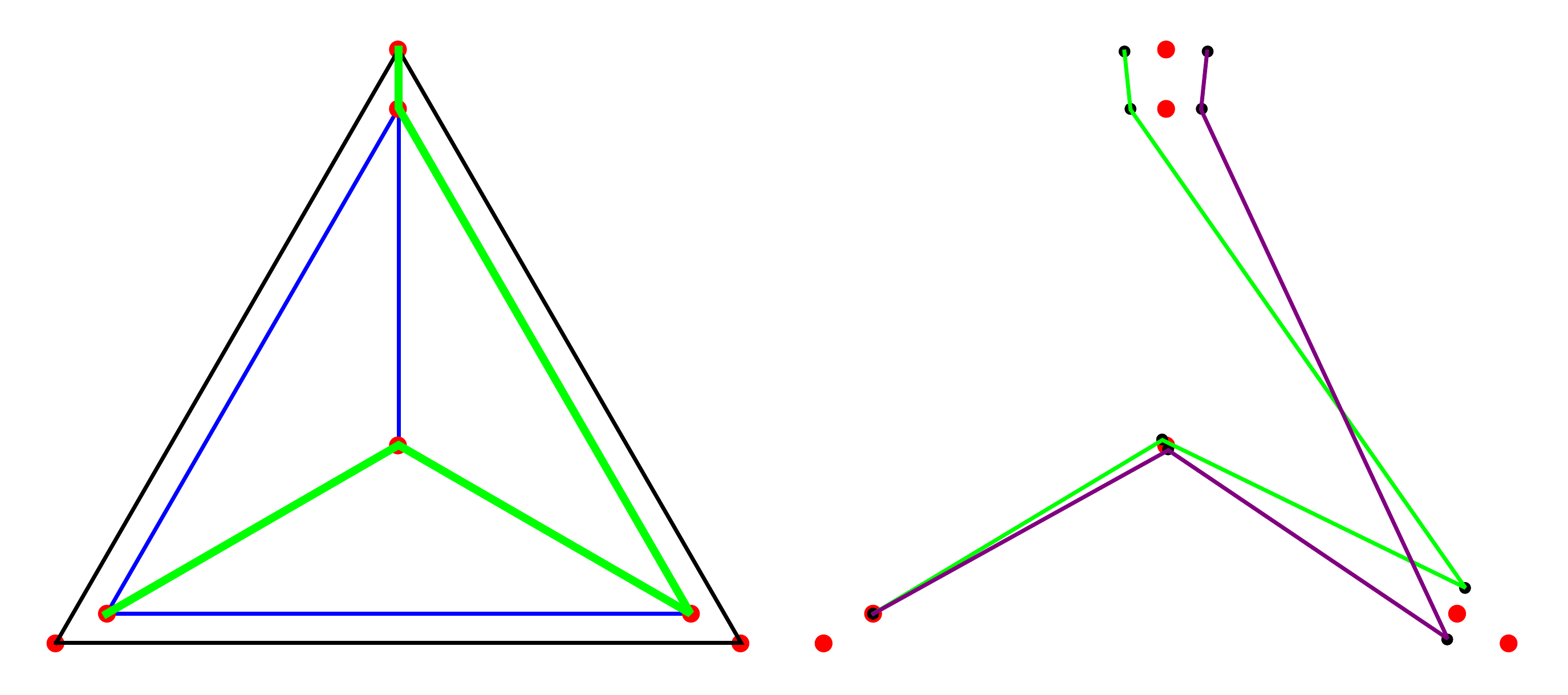}
\caption{%
Left: path $\g = (a_1, c, a_2, a_3, b_3) = acaab$.
Right: $R$ and $L$ developed.
$\{\o_a,\o_c\}=\{5^\circ,10^\circ\}$.}
\figlab{Path_acaa_linear}
\end{figure}
Here $\g$ is shown following straight segments between vertices, and
the developments overlap substantially.
But as per item~5 above, $\g$ can follow potentially any 
(non-self-intersecting) curve between vertices. 
However, the developed images of the vertices are independent of the
shape of the path
between vertices, a condition we exploit in the proof.
So once the combinatorial structure of the cut $\g$ is fixed, the developed
locations of the vertex images are determined.
We will continue to use $\{\o_a,\o_c\}=\{5^\circ,10^\circ\}$ for illustration,
although any smaller curvatures also work in the proofs.

\subsection{Radial Monotonicity: Intuition}
\seclab{RM}
Before beginning the proof details, we provide the intuition behind it.
That intuition depends on the notion of a ``radially monotone'' curve,
a concept used in~\cite{o-ucprm-16} and~\cite{o-eunfcc-17}.
A directed polygonal chain $P$ in the plane with vertices $u_1, u_2, \ldots, u_k$ is
\emph{radially monotone with respect to $u_1$} if the distance from $u_1$ to
every point $p \in P$ increases monotonically as $p$ moves out along the chain.
$P$ is \emph{radially monotone} if it is radially monotone with respect to each
vertex $u_i$: concentric circles centered on each $u_i$ are crossed just
once by the chain beyond $u_i$.

If both the $R$ and $L$ developments are radially monotone, then $L$ and $R$
do not intersect except at their common ``root'' vertex, a fact proved in the
cited papers.\footnote{
There are some curvature bound assumptions to this claim that are not relevant here.}
This suggests that $\g$ should be chosen so that $R$ and $L$ are radially monotone.
However, if $R$ or $L$ or both are not radially monotone, they do not necessarily
overlap: radial monotonicity is sufficient for non-overlap but not necessary.
Nevertheless, striving for radial monotonicity makes sense.
The sharp turns necessary to span the vertices of $\C$ 
(visible in Fig.~\figref{Path_acaa_linear}) should be avoided,
for they violate radial monotonicity.
(Any angle $\angle u_{i-1}, u_i, u_{i+1}$ smaller than $90^\circ$ implies
non-monotonicity at $u_i$ with respect to $u_{i-1}$.)
Avoiding these sharp turns forces $\g$ to exit $\C$ before spanning the vertices.
Although radial monotonicity is not used in the proofs to follow,
it is the intuition behind the proofs.

\subsection{Lemmas~1,2,3,4}
\seclab{Lemmas}
Of the four possible types of $\g$, $acaab$ is the ``closest'' to being
unzippable, so we start with this type.

\begin{lemma}
For sufficiently small $\o_a$, $\o_c$, and $\e$,
any cut-path $\g$ of type $acaab$ must leave and reenter $\C$
to avoid overlap. Therefore, $\C$ cannot be unzipped with
this type of cut-path.
\lemlab{acaab}
\end{lemma}
\begin{proof}
We have already seen in Fig.~\figref{Path_acaa_linear} that straight
connections between the vertices leads to overlap.
Fig.~\figref{RMlogic4frames}(a) repeats the set-up of that figure, with
added notation. Let $R_1,R_2,R_3$ be the portions of the
right development $R$ between vertices, and similarly for $L_i$.
We now imagine that $R_i$ and $L_i$ are arbitrary cuts between their vertex
endpoints. We concentrate on $R_3$ and $L_3$.

From the fact that the images of the vertices, and in particular, $a_2$,
are in their correct developed planar locations, we can derive
constraints on the shape of the $R_3$ and $L_3$ paths.
The shape of $R_i$ determines $L_i$ and vice versa, because for all
non-vertex points of $\g$, $\r+\l = 2\pi$.
Thus $R_i$ and $L_i$ are congruent as curves, but rigidly rotated differently
by the curvatures along $\g$.

There are only two topological possibilities for $R_3$ and $L_3$
to avoid 
crossing earlier portions of $R$ and $L$,
illustrated in Fig.~\figref{RMlogicTopology}.
In~(a) of the figure, $R_3$ passes right of $a''_2$ on its way
counterclockwise to $a'_3$, and 
in~(b), $L_3$ passes right of $a'_2$ on its way
clockwise to $a''_3$.
The situations are analogous in the neighborhood of $a_2$,
and we concentrate only on the former more direct route.

Knowing that $R_3$ passes to the right of $a''_2$ determines the vector 
displacement of the tightest possible prefix $(a'_2,a''_2)$ of $R_3$, but not the shape of that
prefix. This vector displacement forms an effective angle
$\angle c',a'_2,a''_2$ of much larger than the near-$30^\circ$ necessary
to stay on the narrow $\e$-apron.
Fig.~\figref{RMlogic4frames}(a) and~(b) show that this angle is
nearly $70^\circ$, well beyond $30^\circ$.
(And the angle is larger if $R_3$ passes further to the right of $a''_2$.)
This $70^\circ$ turn implies an effective surface angle $\r = 290^\circ$ to the right of $\g$
on $\C$ at $a_2$, ``effective'' because the exact shape of $\g$ is unknown.
The exact angle and length of vector displacement depend on $\{\o_a,\o_c\}$,
but for any given curvatures, we can choose an $\e$ small enough so
that the prefix steps $\g$ exterior to $\C$.
Thus $\g$ must leave $\C$ to
avoid overlap before it completes its tour of the vertices.
Although this proves the lemma, we continue the analysis
below to reveal a deeper structure.
\end{proof}

\begin{figure}[htbp]
\centering
\includegraphics[width=0.75\linewidth]{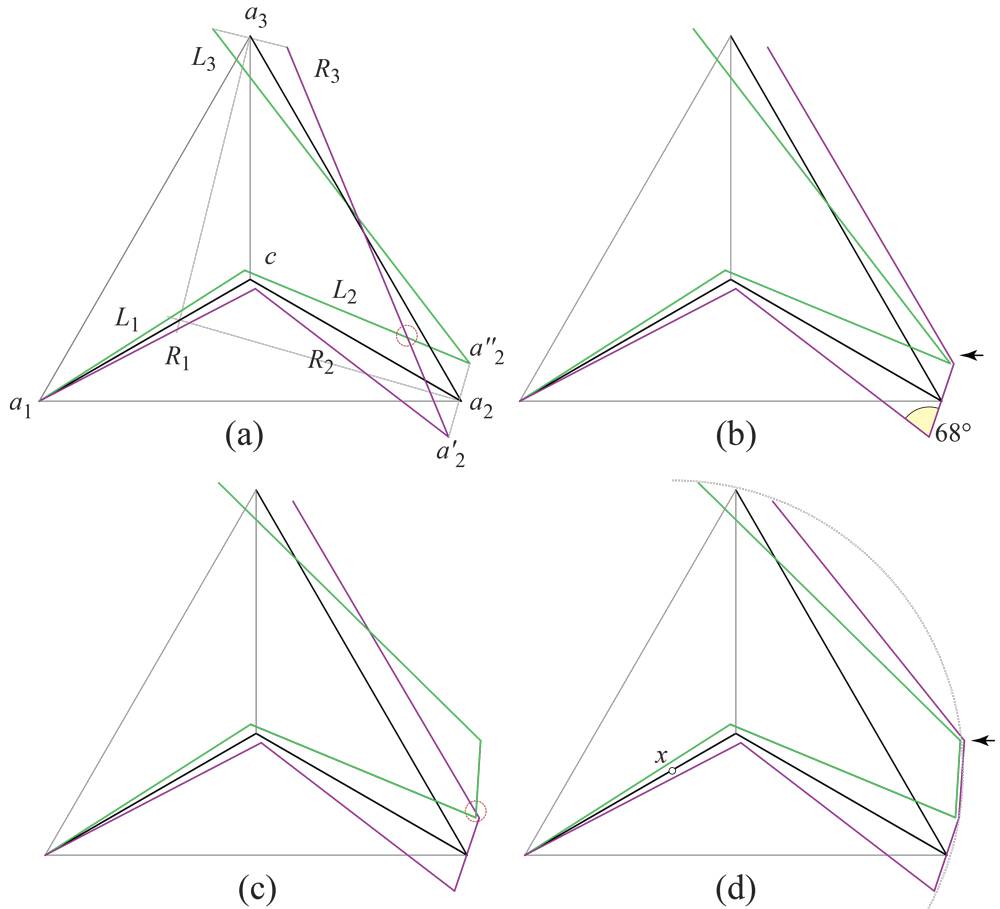}
\caption{Analysis of $\g$ of type  $acaab$.
(a)~Opening at $a_1$ and $c$ causes $R_3 / L_3$ overlap.
(b)~$R_3$ bends around $a''_2$.
(c)~$L_3$ complements $R_3$, which again intersects $L_3$.
(d)~$L_3$ complements $R_3$. $R_3$ is following the arc centered on $x$.}
\figlab{RMlogic4frames}
\end{figure}

\begin{figure}[htbp]
\centering
\includegraphics[width=0.85\linewidth]{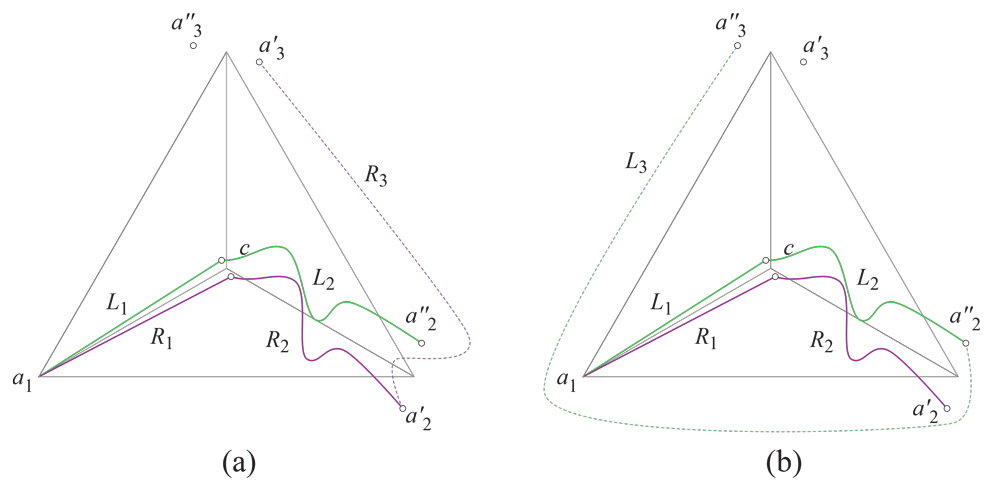}
\caption{Possible paths from $a_2$ to $a_3$.
(a)~$R_3$ skirts $a''_2$.
(b)~$L_3$ skirts $a'_2$.}
\figlab{RMlogicTopology}
\end{figure}

Knowing this constraint just derived on the prefix of $R_3$, we 
know $L_3$ must complement $R_3$ on that prefix,
which leads to Fig.~\figref{RMlogic4frames}(c).
Again there is an overlap intersection further along $R_3$. Altering $R_3$ to again bend around $L_3$
leads to Fig.~\figref{RMlogic4frames}(d).
Continuing this process of incrementally determining constraints on $R_3$ that
are mirrored in $L_3$, 
leads to the conclusion that $R_3$ must follow (or be outside of) the arc of a circle
centered at $x$, where $x$ is the combined center of rotation of
the rotations at $a_1$ and at $c$.
This is why we can be sure the angle at $a_2$ is well beyond $30^\circ$
independent of $\{\o_a,\o_c\}$:
it is determined by (an approximation to) the tangent to this circle.
Finally, $R_3$ and $L_3$ are opened further by the curvature $\o_a$ at $a_2$, which does
not alter the previous analysis. 

Here we pause to discuss the ``combined center of rotation'' just used.
Any pair of rotations about two distinct points is equivalent to a single rotation
about a combined center. In our situation, the two rotations are 
$\o_a$ about $a_1$ and $\o_c$ about $c$.
For small rotations, they are equivalent to a rotation by $\o_a+\o_c$
about the weighted center
$$
x= \frac{ \o_a a_1 + \o_c c }{ \o_a+\o_c} \;.
$$
This point is indicated in Figs.~\figref{RMlogic4frames}(a) and~(d).
This result on combining rotations is proved in both~\cite{o2017addendum} and~\cite{barvinok2017pseudo}
(and likely elsewhere).

The above analysis suggests that $\C$ can be unzipped if the apron were large enough
to include the circle arc that $\g$ must follow from $a_2$ to $a_3$. And indeed
Fig.~\figref{Devel_acaa_arc_5_10} shows that this is true.
An interesting consequence of this unzipping is that, 
even with a small apron, if we close the convex cap $\C$
by adding an equilateral triangle base to form a closed convex polyhedron $\P$,
then $\P$ does have an unzipping. Follow the path shown
in Fig.~\figref{Devel_acaa_arc_5_10}, and complete it by
extending $\g$ to cut $(b_3, b_1, b_2)$, leaving $b_2 b_3$ uncut. Then the
arc illustrated would lie on the unfolding of the base $\triangle b_1 b_2 b_3$.

\begin{figure}[htbp]
\centering
\includegraphics[width=0.75\linewidth]{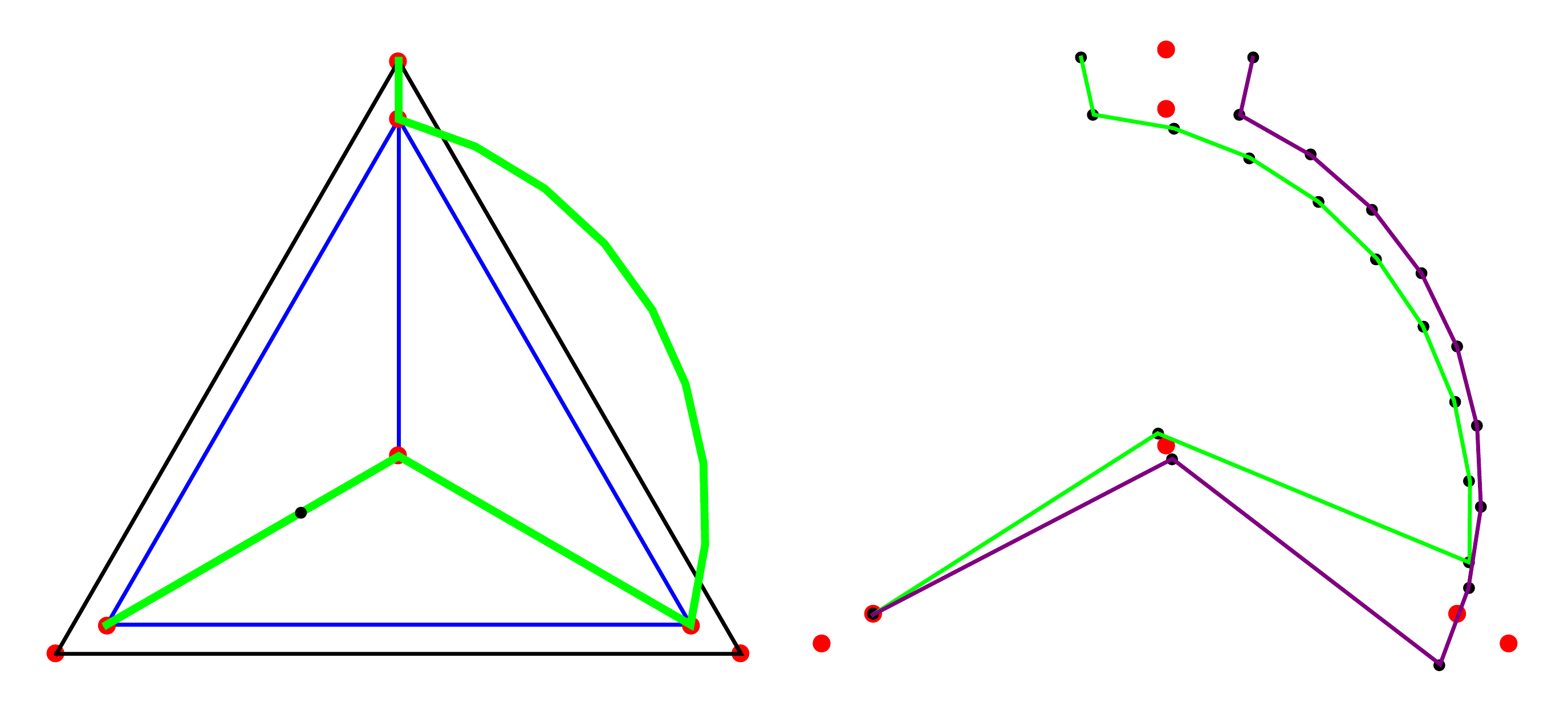}
\caption{An $acaab$ unzipping of $\C$ extending outside $\bcC$.
Right: The $R$ and $L$ developments do not cross.}
\figlab{Devel_acaa_arc_5_10}
\end{figure}

We now turn to the other three types of cut-paths $\g$.
The proof for type $caaab$ is similar to Lemma~\lemref{acaab},
and so will only be sketched.

\begin{lemma}
For sufficiently small $\o_a$, $\o_c$, and $\e$,
any cut-path $\g$ of type $caaab$ must leave and reenter $\C$
to avoid overlap. Therefore, $\C$ cannot be unzipped with
this type of cut-path.
\lemlab{caaab}
\end{lemma}
\begin{proof}
The cut-path with straight segments overlaps at two spots in development,
as shown in Fig.~\figref{Devel_caaa_linear_5_10}
(cf.~Fig.~\figref{Layout_caaa}).
Using the same reasoning as in Lemma~\lemref{acaab},
except that the rotations are centered on $c$ (rather than on both $a_1$
and $c$), leads to the conclusion that $R_2$ and $R_3$ must both
deviate from the $30^\circ$ turn at $a_2$ and the $60^\circ$ turn at $a_3$
needed to stay on an arbitrarily thin apron. In fact, $R_2$ and $R_3$
must follow circle arcs centered on $c$.
Doing so would in fact allow $\C$ to be unzipped if apron were large enough, as shown
in Fig.~\figref{Devel_caaa_arcs_5_10}.
But for an arbitrarily thin $\e$-apron, $\g$ must exit $\C$ before visiting
all vertices, and so cannot be unzipped with this type of cut-path.
\end{proof}
\begin{figure}[htbp]
\centering
\includegraphics[width=0.75\linewidth]{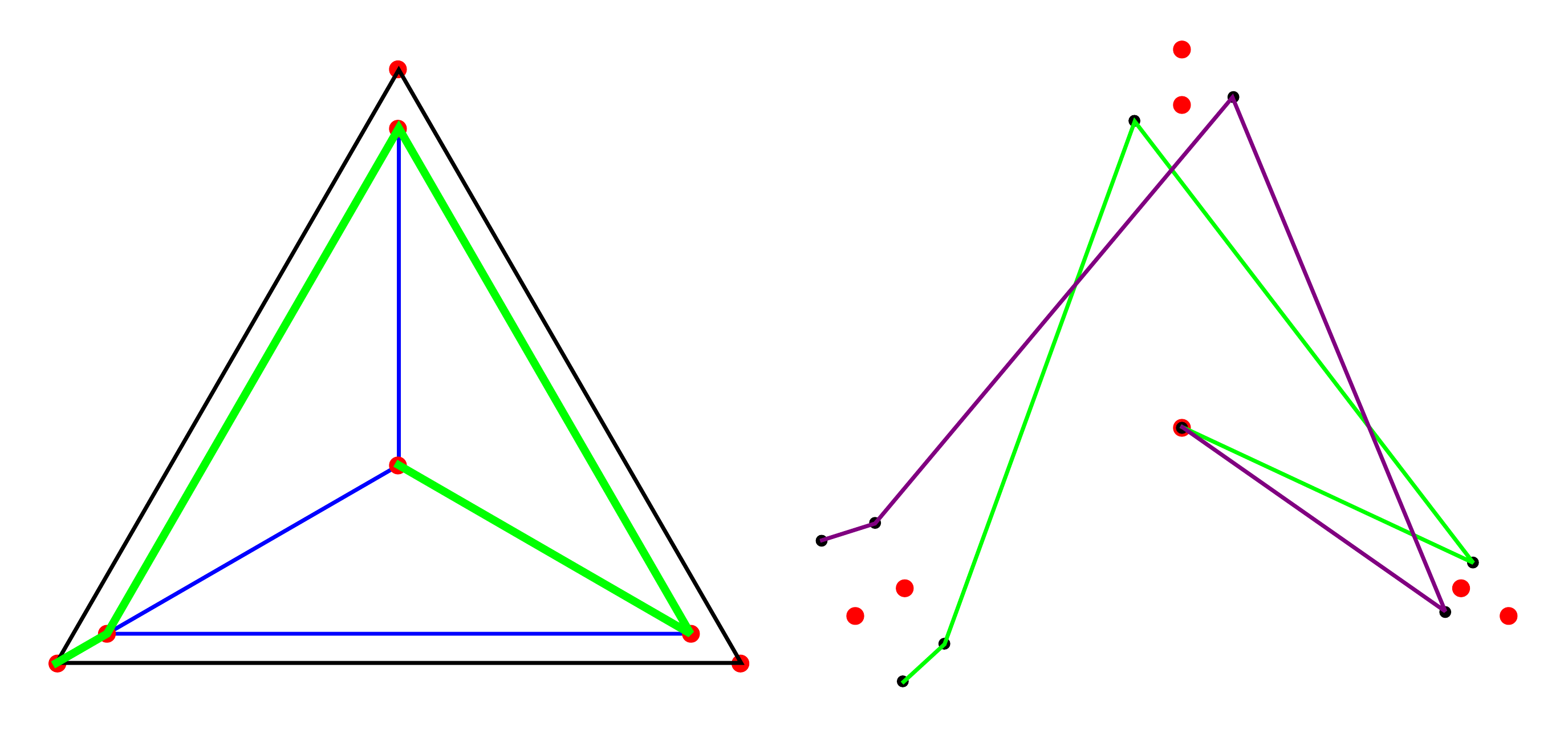}
\caption{The cut-path type $caaab$ leads to overlap with straight segments.}
\figlab{Devel_caaa_linear_5_10}
\end{figure}
\begin{figure}[htbp]
\centering
\includegraphics[width=0.75\linewidth]{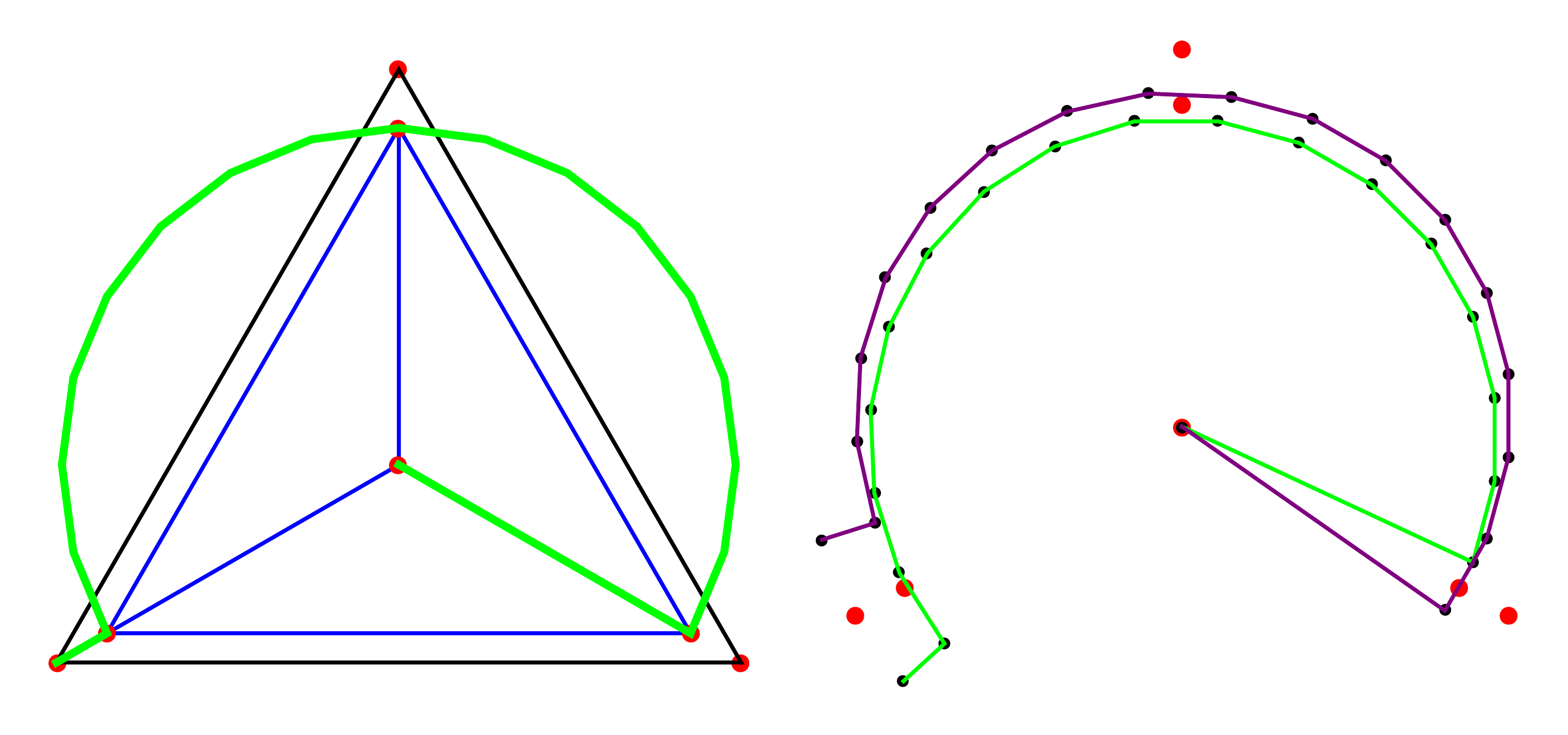}
\caption{An $caaab$ unzipping of $\C$ extending outside $\bcC$.
Right: The $R$ and $L$ developments do not cross.}
\figlab{Devel_caaa_arcs_5_10}
\end{figure}

The third type of cut-path, $aacab$ (Fig.~\figref{Devel_aaca_linear_5_10}),
is different in that not
even following arcs outside of $\C$ would suffice to unzip it without overlap.
\begin{figure}[htbp]
\centering
\includegraphics[width=0.75\linewidth]{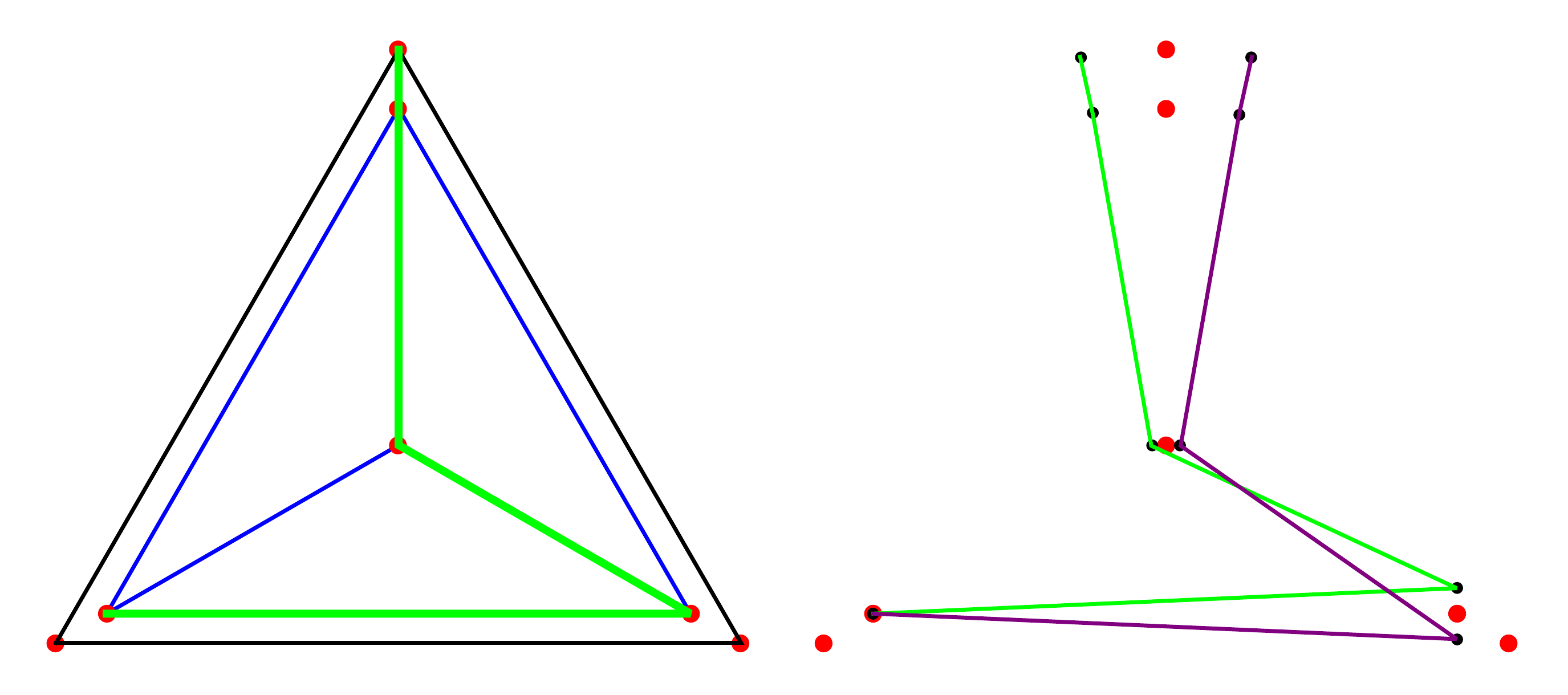}
\caption{The cut-path type $aacab$ leads to $L / R$ overlap.}
\figlab{Devel_aaca_linear_5_10}
\end{figure}

\begin{lemma}
For sufficiently small $\o_a$, $\o_c$, and $\e$,
any cut-path $\g$ of type $aacab$ cannot visit all vertices
without overlap in the development. 
Therefore, $\C$ cannot be unzipped with
this type of cut-path.
\lemlab{aacab}
\end{lemma}
\begin{proof}
We analyze the constraints on $R_2$ in
Fig.~\figref{RMlogic_aaca}.
The rotation at $a_1$ determines the prefix of $R_2$ following the same
reasoning as in Lemma~\lemref{acaab},
and again already in Fig.~\figref{RMlogic_aaca}(b) we have an angle
at $a_2$ much larger than the $30^\circ$ turn required to reach $c$.
This already establishes the cut-path cannot be an unzipping.
But in fact, it is clear that $R_2$ must follow the circle arc shown in 
Fig.~\figref{RMlogic_aaca}(d), centered on $a_1$. Following this arc
makes it impossible for $R_2$ to reach $c$:
the path is forced toward $a_3$ instead.
\end{proof}

\begin{figure}[htbp]
\centering
\includegraphics[width=0.75\linewidth]{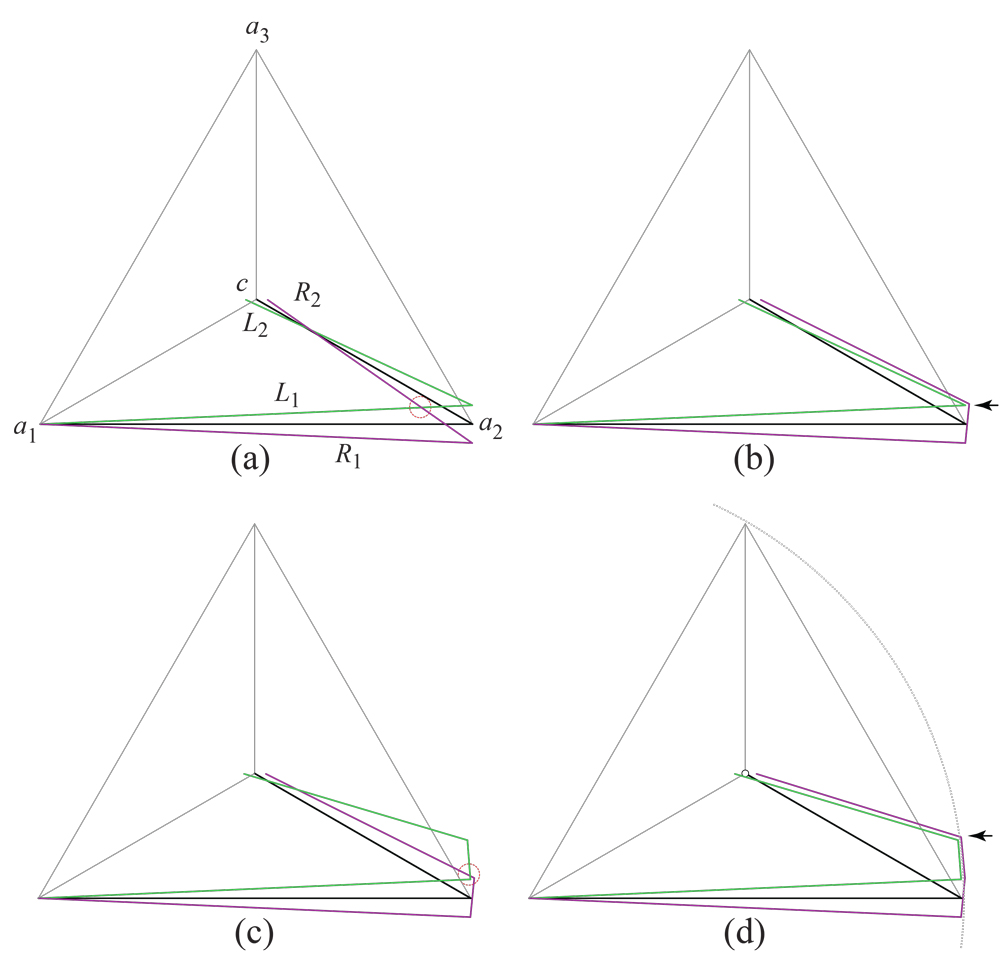}
\caption{Analysis of $\g$ of type  $aacab$.
(a)~Opening at $a_1$  causes $R_2 / L_2$ overlap.
(b)~$R_2$ bends around $a''_2$.
(c)~$L_2$ complements $R_2$, which again intersects $L_2$.
(d)~$L_2$ complements $R_2$. $R_2$ is following the arc centered on $a_1$.}
\figlab{RMlogic_aaca}
\end{figure}

\noindent
The last combinatorial cut-path type, $aaacb$, mixes themes in the
others: first, $\g$ must go outside $\C$, which already establishes there
is no unzipping, and second, even if the apron were large enough, $\g$ cannot
reach $c$. We rely just on the first impediment.

\begin{lemma}
For sufficiently small $\o_a$, $\o_c$, and $\e$,
any cut-path $\g$ of type $aaacb$ cannot visit all vertices
without leaving and re-entering $\C$.
Therefore, $\C$ cannot be unzipped with
this type of cut-path.
\lemlab{aaacb}
\end{lemma}
\begin{proof}
Fig.~\figref{Devel_aaac_linear_5_10} shows there is overlap when $\g$ is composed of straight segments.
By now familiar reasoning, the portion of $\g$ from $a_2$ to $a_3$ must follow a
circular arc centered on $a_1$. This is illustrated in
Fig.~\figref{Devel_aaac_arc_5_10}, and already steps outside
and $\e$-thin apron in the neighborhood of $a_2$, where it
makes an angle of approximately $90^\circ$ rather than the necessary $60^\circ$.
This establishes the claim of the lemma.
\end{proof}

\begin{figure}[htbp]
\centering
\includegraphics[width=0.75\linewidth]{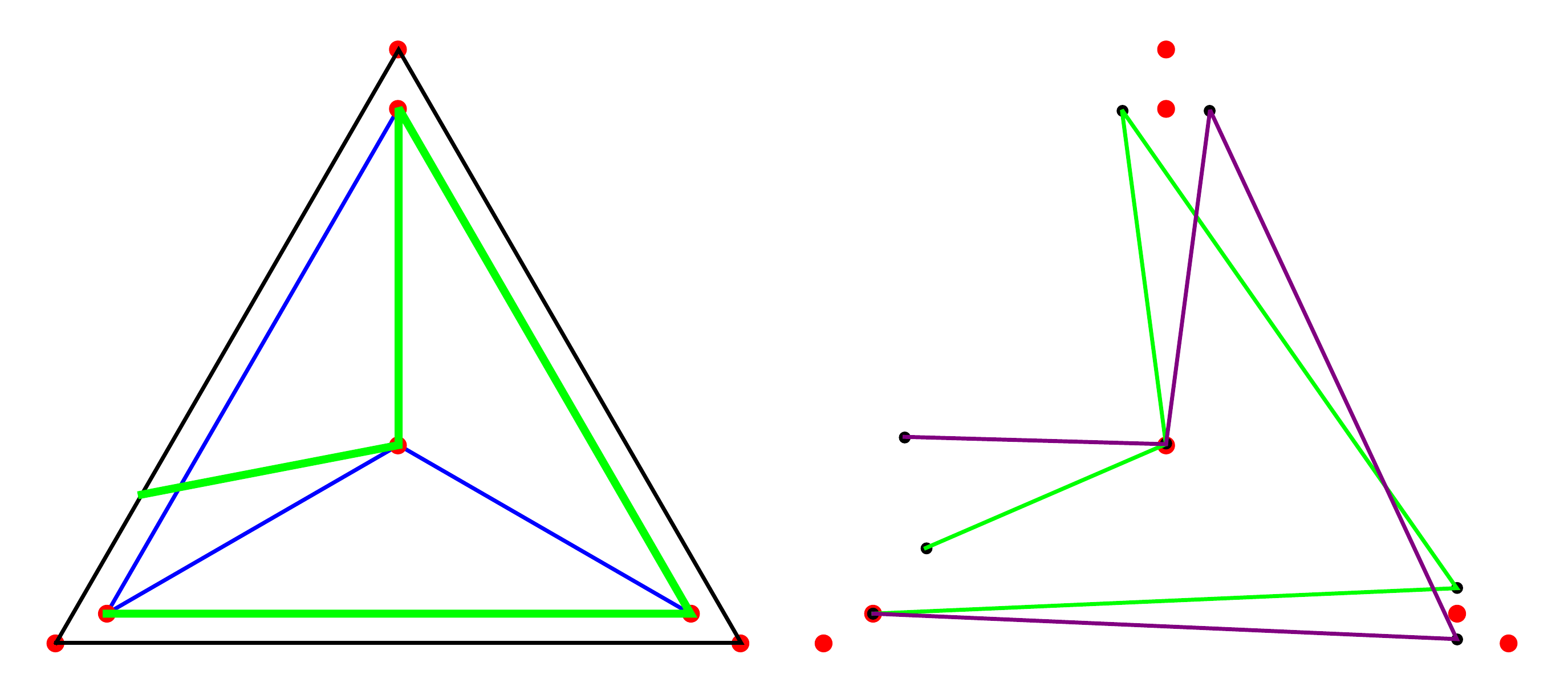}
\caption{The cut-path type $aaacb$ leads to $L / R$ overlap with straight segments.}
\figlab{Devel_aaac_linear_5_10}
\end{figure}

\begin{figure}[htbp]
\centering
\includegraphics[width=0.75\linewidth]{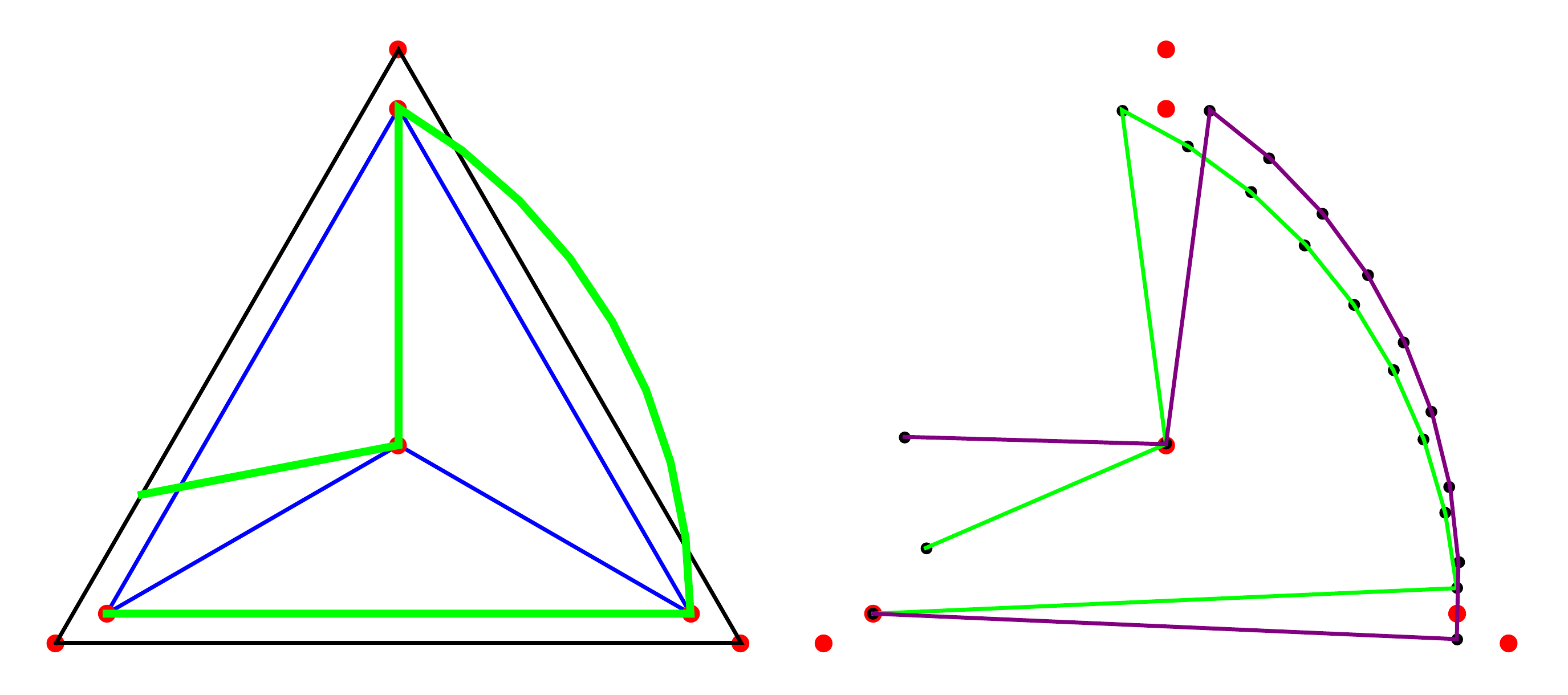}
\caption{The cut-path type $aaacb$ extends outside $\C$ on the $(a_2,a_3)$ arc
(and still $R$ crosses $L$ near $a_3)$.}
\figlab{Devel_aaac_arc_5_10}
\end{figure}


\medskip
\noindent
We restate Theorem~1 in more detail:

\setcounter{theorem}{0}
\begin{theorem}
Convex caps $\C$ with sufficiently small 
$\{\o_a,\o_c, \e\}$, as depicted in Fig.~\figref{ZipCex3D},
have no unzipping: they are un-unzippable.
Thus there are arbitrarily flat convex caps that cannot be unzipped.
\end{theorem}
\begin{proof}
We argued that only four combinatorial types of cut-paths $\g$
are possible on $\C$.
Lemmas~\lemref{acaab}, \lemref{caaab}, \lemref{aacab}, \lemref{aaacb}
established that for sufficiently small
curvatures $\{\o_a,\o_c\}$ and a sufficiently thin $\e$-apron,
each of these cut-path types fails to unzip $\C$.
Because the arguments are independent of the exact values of 
$\{\o_a,\o_c\}$, only requiring a sufficiently small $\e$ to match,
the claim holds for arbitrarily flat convex caps.
\end{proof}

\section{Discussion}
\seclab{Discussion}
It is tempting to hope that the negative result of Theorem~\thmref{ZipCex}
can somehow be used to address the open problem for convex polyhedra.
However, as mentioned earlier (Sec.~\secref{Lemmas}), closing the convex cap in Fig.~\figref{ZipCex3D}
by adding a base creates a polyhedron that can in fact be unzipped.
Perhaps this is not surprising, as the proofs rely crucially
on the fact that $\C$ has a boundary $\bcC$.
I have also explored using several un-unzippable convex caps to tile
a closed convex polyhedron, but so far to no avail.
The open problem from~\cite{lddss-zupc-10} 
quoted in Sec.~\secref{Introduction} remains open.


\newpage 
\bibliographystyle{alpha}
\bibliography{ZipCex}

\newcommand{\etalchar}[1]{$^{#1}$}
\begin{thebibliography}{DDL{\etalchar{+}}10}

\bibitem[BDE{\etalchar{+}}03]{bdekms-upcf-03}
Marshall Bern, Erik~D. Demaine, David Eppstein, Eric Kuo, Andrea Mantler, and
  Jack Snoeyink.
\newblock Ununfoldable polyhedra with convex faces.
\newblock {\em Comput. Geom. Theory Appl.}, 24(2):51--62, 2003.

\bibitem[BG17]{barvinok2017pseudo}
Nicholas Barvinok and Mohammad Ghomi.
\newblock Pseudo-edge unfoldings of convex polyhedra.
\newblock {\em arXiv:1709.04944}, 2017.
\newblock \url{https://arxiv.org/abs/1709.04944}.

\bibitem[DDL{\etalchar{+}}10]{lddss-zupc-10}
Erik Demaine, Martin Demaine, Anna Lubiw, Arlo Shallit, and Jonah Shallit.
\newblock Zipper unfoldings of polyhedral complexes.
\newblock In {\em Proc. 22nd Canad. Conf. Comput. Geom.}, pages 219--222,
  August 2010.

\bibitem[DO07]{do-gfalop-07}
Erik~D. Demaine and Joseph O'Rourke.
\newblock {\em Geometric Folding Algorithms: Linkages, Origami, Polyhedra}.
\newblock Cambridge University Press, July 2007.
\newblock \url{http://www.gfalop.org}.

\bibitem[O'R13]{o-dp-13}
Joseph O'Rourke.
\newblock D{\"u}rer's problem.
\newblock In Marjorie Senechal, editor, {\em Shaping Space: Exploring Polyhedra
  in Nature, Art, and the Geometrical Imagination}, pages 77--86. Springer,
  2013.

\bibitem[O'R15]{o2015spiral}
Joseph O'Rourke.
\newblock Spiral unfoldings of convex polyhedra.
\newblock {\em arXiv:1509.00321}, 2015.
\newblock \url{https://arxiv.org/abs/1509.00321}.

\bibitem[O'R16]{o-ucprm-16}
Joseph O'Rourke.
\newblock Unfolding convex polyhedra via radially monotone cut trees.
\newblock {\em arXiv:1607.07421}, 2016.
\newblock \url{http://arxiv.org/abs/1607.07421}.

\bibitem[O'R17a]{o2017addendum}
Joseph O'Rourke.
\newblock Addendum to: Edge-unfolding nearly flat convex caps.
\newblock {\em arXiv:1709.02433}, 2017.
\newblock \url{https://arxiv.org/abs/1709.02433}.

\bibitem[O'R17b]{o-eunfcc-17}
Joseph O'Rourke.
\newblock Edge-unfolding nearly flat convex caps.
\newblock {\em arXiv:1707.01006v2}, 2017.
\newblock \url{http://arxiv.org/abs/1707.01006}.

\bibitem[She75]{s-cpcn-75}
Geoffrey~C. Shephard.
\newblock Convex polytopes with convex nets.
\newblock {\em Math. Proc. Camb. Phil. Soc.}, 78:389--403, 1975.

\end{thebibliography}
\end{document}